%% file: main.tex
\documentclass[letterpaper, 10 pt, conference]{ieeeconf}
\IEEEoverridecommandlockouts                              
\usepackage{graphicx}    
\usepackage{amsmath}
\usepackage{amssymb}
\usepackage{bm}
\usepackage{mathtools}
\usepackage{color,soul}
\usepackage{url}
\usepackage{comment}
\usepackage{algpseudocode}
\usepackage{algorithm}
\usepackage{mathdots}
\usepackage{enumerate}
\usepackage{txfonts}
\usepackage{microtype}
\usepackage{etoolbox}
\usepackage[stable]{footmisc}
\usepackage{flushend}

\newtoggle{extended}
\toggletrue{extended} 

\newcommand{\extra}[1]{%
	\iftoggle{extended}{#1}{}%
}
\newcommand{\altextra}[2]{%
	\iftoggle{extended}{#1}{#2}%
}

\usepackage[dvipsnames]{xcolor}
\usepackage{tikz}
\usepackage{pgfplots}
\input{tikzsettings.tex}

\newtheorem{assum}{Assumption}
\newtheorem{defn}{Definition}
\newtheorem{rem}{Remark}
\newtheorem{prop}{Proposition}
\newtheorem{lem}{Lemma}
\newtheorem{thm}{Theorem}

\input{notation.tex}

\title{\LARGE \bf
	Output behavior equivalence and simultaneous 
	\\ subspace identification of systems and faults}

\author{Gabriel de Albuquerquer Gleizer
	\thanks{
		G.~A.~Gleizer is with the Delft Center for Systems and Control,
			Delft Technical University, 2628 CD Delft, The Netherlands
			{\tt\small g.gleizer@tudelft.nl}}%
}

\begin{document}
	
	\maketitle
	\thispagestyle{empty}
	\pagestyle{empty}
	\begin{abstract}
		We address the problem of identifying a system subject to additive faults, while simultaneously reconstructing the fault signal via subspace methods. We do not require nominal data for the identification, neither do we impose any assumption on the class of faults, e.g., sensor or actuator faults. We show that, under mild assumptions on the fault signal, standard PI-MOESP can recover the system matrices associated to the input--output subsystem. Then we introduce the concept of output behavior equivalence, which characterizes systems with the same output behavior set, and present a method to establish this equivalence from system matrices. Finally, we show how to estimate from data the complete set of fault matrices for which there exist a fault signal with minimal dimension that explains the data.
	\end{abstract}
	

    \input{introduction.tex}

    \input{problem.tex}

    \input{result.tex}

	\input{numerical.tex}

	\section{CONCLUSION}
	In this work we have presented a subspace method to identify an LTI system from faulty input--output data. The core of this contribution is the method to recover the fault matrices $\Fm, \Gm$ without any structural assumptions on them. This problem has no unique solution, and the presence of transmission zeros from fault to output can dramatically increase the set of systems that explain the data from a faulty experiment. What is the precise relation is still under research. In addition, the method is not always robust to noise, so the modifications to improve robustness also merits further research. The notion of output behavior equivalence is interesting from the perspective of realization theory, and has potential for use in other applications, such as actuator selection. Future work will be dedicated to extending this notion to other classes of systems.
	
	\section*{ACKNOWLEDGMENT}
	
	The author thanks Prof.~Michel Verhaegen for the valuable discussions and suggestions.
	
	\input{appendix.tex}

	\bibliographystyle{ieeetr} 
	\bibliography{mybib}

\end{document}

%% file: tikzsettings.tex
\pgfplotsset{compat=newest}
\usetikzlibrary{automata}
\usetikzlibrary{arrows, arrows.meta, decorations.markings, patterns}
\usepgfplotslibrary{patchplots, groupplots}
\usetikzlibrary{shapes, positioning, fit, backgrounds}
\usepgfplotslibrary{fillbetween}

\usetikzlibrary{overlay-beamer-styles}
\usetikzlibrary{calc}

\usepackage{calc}

\definecolor{tudcyan}{RGB}{0,166,214}
\definecolor{tudmagenta}{RGB}{109,23,127}
\definecolor{tudpurple}{RGB}{29,28,115}
\definecolor{tudgraygreen}{RGB}{107,134,137}

\colorlet{lighttudcyan}{tudcyan!20}
\colorlet{lighttudmagenta}{tudmagenta!20}

\newlength{\hatchspread}
\newlength{\hatchthickness}
\newlength{\hatchshift}
\newcommand{\hatchcolor}{}
\tikzset{hatchspread/.code={\setlength{\hatchspread}{#1}},
	hatchthickness/.code={\setlength{\hatchthickness}{#1}},
	hatchshift/.code={\setlength{\hatchshift}{#1}},
	hatchcolor/.code={\renewcommand{\hatchcolor}{#1}}}
\tikzset{hatchspread=7pt,
	hatchthickness=0.5pt,
	hatchshift=0pt,
	hatchcolor=black}
\pgfdeclarepatternformonly[\hatchspread,\hatchthickness,\hatchshift,\hatchcolor]
{custom north west lines}
{\pgfqpoint{\dimexpr-2\hatchthickness}{\dimexpr-2\hatchthickness}}
{\pgfqpoint{\dimexpr\hatchspread+2\hatchthickness}{\dimexpr\hatchspread+2\hatchthickness}}
{\pgfqpoint{\dimexpr\hatchspread}{\dimexpr\hatchspread}}
{
	\pgfsetlinewidth{\hatchthickness}
	\pgfpathmoveto{\pgfqpoint{0pt}{\dimexpr\hatchspread+\hatchshift}}
	\pgfpathlineto{\pgfqpoint{\dimexpr\hatchspread+0.15pt+\hatchshift}{-0.15pt}}
	\ifdim \hatchshift > 0pt
	\pgfpathmoveto{\pgfqpoint{0pt}{\hatchshift}}
	\pgfpathlineto{\pgfqpoint{\dimexpr0.15pt+\hatchshift}{-0.15pt}}
	\fi
	\pgfsetstrokecolor{\hatchcolor}
	\pgfusepath{stroke}
}

\def\centerarc[#1](#2)(#3:#4:#5)
{ \draw[#1] ($(#2)+({#5*cos(#3)},{#5*sin(#3)})$) arc (#3:#4:#5); }

%% file: notation.tex
\newcommand*{\tran}{^{\mkern-1.5mu\mathrm{T}}\!}  
\newcommand*{\kernel}{^{\mathcal{N}}\!}
\newcommand*{\m}[1]{\begin{bmatrix} #1 \end{bmatrix}}
\newenvironment{lbmatrix}[1]
{\left[\array{@{}*{#1}{c}@{}}}
{\endarray\right]}
\newcommand*{\mm}[2]{\begin{lbmatrix}{#1} #2 \end{lbmatrix}}

\newcommand*{\bsm}[1]{\begin{bsmallmatrix} #1 \end{bsmallmatrix}}

\DeclareMathOperator{\rank}{rank}

\DeclareMathOperator{\E}{\mathbb{E}}

\def\norm[#1]{\left|#1\right|}
\def\lnorm[#1]{\left\|#1\right\|}
\def\shortnorm[#1]{|#1|}
\def\snorm[#1]{\left\|#1\right\|_\mathrm{S}}
\def\shortsnorm[#1]{\|#1\|_\mathrm{S}}
\def\fnorm[#1]{\left\|#1\right\|_\mathrm{F}}
\def\shortfnorm[#1]{\|#1\|_\mathrm{F}}

\def\q{\mathfrak{q}}
\def\z{\mathrm{z}}

\def\No{\mathbb{N}_{0}}
\def\N{\mathbb{N}}
\def\R{\mathbb{R}}

\def\C{\mathbb{C}}

\def\Rs{\mathcal{R}}
\def\Bs{\mathcal{B}}  

\def\nup{{n_{\mathrm{u}}}}
\def\nx{{n_{\mathrm{x}}}}
\def\ny{{n_{\mathrm{y}}}}
\def\nw{{n_{\mathrm{w}}}}

\def\nv{{n_{\mathrm{v}}}}

\def\nz{{n_{\mathrm{z}}}}

\def\xv{\bm{x}}
\def\yv{\bm{y}}
\def\uv{\bm{u}}
\def\wv{\bm{w}}
\def\vv{\bm{v}}
\def\rv{\bm{r}}
\def\xiv{\bm{\xi}}

\def\zv{\bm{z}}

\def\Am{\bm{A}}
\def\Bm{\bm{B}}
\def\Cm{\bm{C}}
\def\Dm{\bm{D}}

\def\I{\textbf{I}}
\def\O{\bm{0}}
\def\Fm{\bm{F}}
\def\Gm{\bm{G}}
\def\Hm{\bm{H}}
\def\Mm{\bm{M}}
\def\Nm{\bm{N}}
\def\Pm{\bm{P}}
\def\Qm{\bm{Q}}

\def\Tm{\bm{T}}
\def\Rm{\bm{R}}

\def\Xm{\bm{X}}
\def\Ym{\bm{Y}}
\def\Um{\bm{U}}
\def\Wm{\bm{W}}
\def\Vm{\bm{V}}
\def\Zm{\bm{Z}}

\def\Lm{\bm{L}}
\def\Om{\bm{O}}
\def\Jm{\bm{J}}


%% file: introduction.tex
\section{INTRODUCTION}

    In many applications, it is difficult to identify a system in a controlled environment. For example, when faulty behavior appears, it may be impractical or economically undesirable to remove the system from the field to perform a new identification experiment. For example, digital twins must enable tasks such as re-identification of system parameters, uncovering new dynamics, and detecting and reconstructing obtrusive signals that are affecting the system performance, all executed without removing the system from operation. %
    Thus, an important problem is to simultaneously estimate the matrices that describe the system model and its external disturbances; for the latter, it is desirable to both learn how they affect the system and what is the disturbance signal that explains the data. %
    Technically, we consider \emph{simultaneous identification and input reconstruction} of LTI systems of the form
    \begin{equation*}
        \begin{aligned}
            \xv(k+1) &= \Am\xv(k) + \Bm\uv(k) + \Fm\vv(k),\\
            \yv(k) &= \Cm\xv(k) + \Dm\uv(k) + \Gm\vv(k),
        \end{aligned}\label{eq:system_intro}
    \end{equation*}
    where $\xv(k), \uv(k), \yv(k),$ and $\vv(k)$ represent the unknown state, known input, known output, and unknown fault at time $k$, respectively. That is, the goal is to estimate the matrices $\Am, \Bm, \Cm, \Dm, \Fm, \Gm$ from a series of values of $\uv$ and $\yv$, as well as the corresponding trajectories of the state $\xv$ and fault $\vv$, which is also called unknown input or disturbance in the literature. We fix the word \emph{fault} for $\vv$ in this note, but whether it is to be considered a fault depends on the application.
    
    In system identification, the fault has usually a probabilistic nature; e.g., it is generally assumed to be wide-sense stationary (WSS), which enables very efficient and well-established methods (see, e.g., \cite{verhaegen2007filtering}). Unfortunately, many disturbances cannot be assumed to be WSS unless very long time horizons are taken: this is the case of environmental conditions such as temperature and wind, and the influence of other subsystems that may be connected to the system of interest. In the present work, we impose no probabilistic assumption on $\vv.$ %
    Most of the literature of data-driven fault estimation restricts the class of faults to either sensor  ($\Fm=\O, \Gm=\Dm$), actuator faults ($\Fm=\Bm, \Gm=\O$) {\cite{naderi2017data, sheikhi2024kernel}} {or load disturbances  ($\Fm=\I, \Gm=\O$) \cite{liu2020subspace}.} 
    Moreover, the identification is generally a two-step process, where in the first step a fault-free dataset is obtained, which allows identification of the system either through {geometric methods \cite{marro2010unknown}}, system matrices and/or Markov parameters {\cite{dong2011identification},} \cite{wan2017fault}, basis functions \cite{sheikhi2024kernel}, or behavioral subspaces \cite{yan2025secure}; {see \cite{ding2014data} for an overview.} Once a model is (implicitly) identified, the fault matrices can be established from the system matrices, and input reconstruction becomes fundamentally a problem of system inversion or, under stricter assumptions on the fault signal, input tracking \cite{van_der_ploeg_multiple_2022, dong2025robust}.  {Fewer works are dedicated to simultaneous identification of systems and external disturbances. Most require some extra assumption on the fault signal, e.g., it is constant, slow-varying or periodic \cite{hou2018recursive, liu2020subspace, hou2021subspace}.} One of the few papers that consider simultaneous identification and reconstruction of unknown inputs {without any requirement on their predictability} is \cite{palanthandalam2009subspace}, and as such it is the closest this work. %
    
    Another related area is blind identification \cite{xu1995least}. In this setting, the system matrices and the input are unknown, and dropping stationarity assumptions on the input makes the problem extremely challenging. For this reason, extra assumptions have to be imposed, e.g., the input is constant \cite{markovsky2015application}, or it is the output of an autonomous LTI system \cite[Sec.~8.5.1]{markovsky2018low}. Our problem is easier than general blind identification, because we assume to be in possession of  the input $\uv,$ with which the system can be excited. This allows us to impose very little assumptions on the input signal itself.

    As stated earlier, the closest work to the present one is \cite{palanthandalam2009subspace}. Therein, the system is assumed to be strongly input observable \cite{hou1998input}, and the unknown input dimension is known. We drop these assumptions to understand two fundamental problems: (i) What is the smallest unknown input dimension that explains the data? (ii) What can be learned if strong input observability does not hold? In other words, what is the set of systems and corresponding inputs that could explain the obtained output? %
    In addition, \cite{palanthandalam2009subspace} rely on strong input and state observability to estimate the nominal system matrices $(\Am,\Bm,\Cm,\Dm)$. We show that this is sufficient but not necessary. Our contributions are the following: 
    (i) Under relatively mild assumptions on the fault signal, using zero-mean ergodic inputs allows one to use standard PI-MOESP \cite{verhaegen1993subspace} to consistently identify the nominal system matrices $\Am,\Bm,\Cm,\Dm$, without the need for strong input and state observability.
     (ii) We introduce the notion of \emph{output behavioral equivalence} to compare different systems that can generate the same output trajectories. %
     (iii) Given identified system matrices, the problem can be reduced to simultaneous (blind) identification of fault matrices $\Fm, \Gm$, the initial state $\xv(0)$ and the fault signal $\vv$. We show how to determine the minimal input dimension that explains the output, and provide an algorithm to identify the set of matrices $\Fm, \Gm$ that can be used for reconstruction of compatible values of $\xv(0)$ and $\vv$.
     {Finally, we provide the Matlab code to execute the proposed methods and reproduce the results of this note in \url{https://github.com/ggleizer/fault_id}.}%

%% file: problem.tex
\section{PRELIMINARIES AND PROBLEM FORMULATION}\label{sec:prob}

    \subsection{Mathematical notation}
    
	We denote by $\No$ the set of natural numbers including zero, $\N \coloneqq \No \setminus \{0\}$ and by $\R$ and $\C$ the sets of real and complex numbers, respectively. %
	We denote by $\Rs(\Am)$ the column space (range) of the matrix $\Am \in \R^{n\times m}$, by $\Am\tran$ its transpose, and by ${\lnorm[\Am]}_*$ its nuclear norm. %
    When matrix $\Am$ is left (resp.~right) invertible we denote by $\Am^{-L}$ (resp.~$\Am^{-R}$) any left (resp.~right) inverse of $\Am$. For indexed vectors or matrices $\Am_i$, the matrix $\Am_{k:l},$ $l>k$, is built by vertically stacking $\Am_i$ from $k$ to $l$.

    We denote the space of polynomial matrices with real coefficients with $n$ rows and $m$ columns by $\R[\cdot]^{n \times m}$. %
    The value of $\max_{\q\in\C}\rank(\Am(\q))$ is called the \emph{normal rank} of $\Am(\q)$, and it is equal to $\rank(\Am(\q))$ for almost all $\q\in\C$.
	
	\subsection{System model and problem statement}

    We consider discrete-time linear time-invariant systems of the form
    \begin{equation}
        \begin{aligned}
            \xv(k+1) &= \Am\xv(k) + \Bm\uv(k) + \Fm\vv(k),\\
            \yv(k) &= \Cm\xv(k) + \Dm\uv(k) + \Gm\vv(k) + \wv(k),
        \end{aligned}\label{eq:system}
    \end{equation}
    where $\xv(k) \in \R^\nx, \uv(k)\in \R^\nup, \yv(k)\in \R^\ny, \vv(k)\in \R^\nv,$ and $\wv(k)\in \R^\nw$ represent the state, input, output, fault, and noise values at time $k$, respectively. All matrices are unknown, but the pairs $(\Am,\Cm)$ and $(\Am,\Bm)$ are assumed to be observable and controllable, respectively, and $\Am$ is assumed to be Schur, i.e., all its eigenvalues are strictly inside the complex unit circle. Input and output signals are assumed to be available data, but states, faults and noise are not; moreover, one can freely choose the input $\uv(k)$.
    The distinction between the unknown signals $\vv(k)$ and $\wv(k)$ is of nature: $\wv(k)$ is a realization of a second-order ergodic, WSS random process  with zero mean; it represents colored noise, which can include measurement and process noise. In contrast, no probabilistic assumption is made on $\vv$. This means that $\vv$ can be generated adversarially --- although with limited powers to the adversary---, have time-varying average, or be the result of a nonlinear phenomenon. A final distinction is that we are interested in recovering the signal $\vv$, whereas this is not the case for $\wv$.

    {\bf Problem statement.} Our goal is to determine under which conditions it is possible to identify the system matrices and recover (a segment of) the signal $\vv$. We want a constructive answer to such conditions; i.e., whenever they hold, we provide an algorithm to perform such a task. Whenever the matrices are not uniquely identifiable, we propose to find the set of solutions to the underdetermined problem.

%% file: result.tex
\section{IDENTIFICATION OF INPUT-OUTPUT SUBSYSTEM}\label{sec:id}

We follow a subspace approach to the identification process. We start by applying a standard reformulation of \eqref{eq:system} so that the fault residual appears directly at the output:
 \begin{equation}
	\begin{aligned}
		\tilde{\xv}(k+1) &= \Am\tilde{\xv}(k) + \Bm\uv(k),\\
		\yv(k) &= \Cm\tilde{\xv}(k) + \Dm\uv(k) + \rv(k) + \wv(k),
	\end{aligned}\label{eq:system_fault_output}
\end{equation}
with the residual generating process
 \begin{equation}
	\begin{aligned}
		\xiv(k+1) &= \Am\xiv(k) + \Fm\vv(k),\\
		\rv(k) &= \Cm\xiv(k) + \Gm\vv(k),
	\end{aligned}\label{eq:residual_system}
\end{equation}
so that $\xiv = \xv - \tilde{\xv}.$ Now, we write the data equation for \eqref{eq:system_fault_output}:
\begin{equation}
    \Ym_{2s} = \Om_{2s}\tilde{\Xm}_{T} + \Tm_{2s}\Um_{2s,T} + \Rm_{2s,T} + \Wm_{2s,T}, \label{eq:data_equation}
\end{equation} 
where
\begin{gather*}
	\Om_{2s} \coloneqq \m{\Cm \\ \Cm\Am \\ \vdots \\ \Cm\Am^{2s-1}}, \quad \Tm_{2s} \coloneqq \m{\Dm \\ \Cm\Bm & \Dm \\ \vdots & & \ddots \\ \Cm\Am^{2s-2}\Bm & \cdots & \Cm\Bm & \Dm}, \\
	\tilde{\Xm}_T = \m{\tilde{\xv}(0) & \tilde{\xv}(1) & \cdots & \tilde{\xv}(T)}, \\
	\Um_{2s,T} = \m{\uv(0) & \uv(1) & \cdots & \uv(T-1) \\ \uv(1) & \uv(2) & \cdots & \uv(T) \\  \vdots & \vdots & \ddots & \vdots \\ \uv(2s-1) & \uv(2s) & \cdots & \uv(T+2s-2)},
\end{gather*}
with  $\Rm_{2s,T}$ and $\Wm_{2s,T}$ being built in block-Hankel form as $ \Um_{2s,T}$ from $\rv$ and $\wv$, respectively; $T$ is the number of data samples and $s$ is a window parameter satisfying $T > 2s > 2\nx$ so that the matrix $\Um_{2s,T}$ has more columns than rows.

We start with the following assumptions on the input and fault signals:
\begin{assum}
	\label{ass:ergodic_input}
The input $\uv(k)$ is zero-mean \emph{wide-sense ergodic}, i.e.,  $\lim_{n\to\infty}\frac{1}{n}\sum_{k=0}^n\uv(k) = 0$ and the corresponding matrix \[\lim_{n\to\infty}\frac{1}{n}\Um_{s,N}\Um_{s,N}\tran\] is full rank.
\end{assum}

\begin{assum}
	\label{ass:cross_covariance_zero}
The input and fault signals $\uv$ and $\vv$ satisfy, for all $l \in \N_0,$
\begin{equation*}
    \lim_{n\to\infty}\frac{1}{n}\sum_{k=0}^n\uv(k-l)\vv(k)\tran = \O.
\end{equation*}
\end{assum}

The first assumption is an experiment design assumption, and similar to many existing assumptions in the system identification literature. The second assumption is more restrictive as it imposes conditions on the fault signal, but still rather mild. Under Assumption \ref{ass:ergodic_input}, it holds if, e.g., $\vv(k)$ is bounded for all $k$ (by properties of limit of the product) and is uncorrelated with past values of $\uv$. For this last condition to hold, we essentially impose that the process that generates has no \emph{linear} feedback component \emph{either on $\yv, \xv,$ or $\uv$.} It is well known that the identification problem has no unique solution when faults can include linear feedback terms, see, e.g., \cite{verhaegen2022data}. In practice, if the fault signal includes a linear feedback component, then the identified system will have this feedback component \emph{incorporated} in the system matrices. As $\uv$ does not enter \eqref{eq:residual_system}, we have that $ \lim_{n\to\infty}\frac{1}{n}\sum_{k=0}^n\uv(k-l)\vv(k)\tran = \O$ implies $ \lim_{n\to\infty}\frac{1}{n}\sum_{k=0}^n\uv(k-l)\rv(k)\tran = \O.$

Finally, the following standard assumption is given:
\begin{assum}
	\label{ass:noise}
    The noise $\wv$ is a zero-mean wide-sense stationary process and is uncorrelated with the input, i.e., $\E[\wv(k)\uv(k-l)\tran] = \O$ for all lags $l\geq 0$.
\end{assum}

With Assumptions \ref{ass:ergodic_input}, \ref{ass:cross_covariance_zero}, we have that $\lim_{N\to\infty}\frac{1}{N}\Vm_{s,N}\Um_{s,N}\tran = 0$ the same way $\lim_{N\to\infty}\frac{1}{N}\Wm_{s,N}\Um_{s,N}\tran = 0$ (by Assumption \ref{ass:noise}). Therefore, the same conditions for the PI-MOESP method \cite{verhaegen1993subspace} to provide consistent estimates of $\Am,\Bm,\Cm,\Dm$ apply here, and we have the following result:
\begin{prop}
Let Assumptions \ref{ass:ergodic_input}--\ref{ass:noise} hold. Then, given data sequences $\yv(k),$ $\uv(k),$ the PI-MOESP method as described in \cite{verhaegen1993subspace},\cite[Section 9.5]{verhaegen2007filtering} provides a consistent estimate of the matrices $\Am,\Bm,\Cm,\Dm$.    
\end{prop}

\begin{rem}
    This result may give the false impression that one can generally treat faults as noise. However, as a counterargument, the PO-MOESP \cite[Section 9.6.1]{verhaegen2007filtering} method cannot be used. PO-MOESP further requires that $\E[\yv(k-l)\vv(k)\tran] = \O$, which does not generally hold. E.g., if $\vv$ is not zero-mean and ($\Am,\Fm$) is controllable, then $\E[\yv\vv\tran] \neq \O$. %
\end{rem}
{\begin{rem}
		More generally, any identification method with consistency guarantees in the colored noise cased should also have consistency results as long as Assumptions \ref{ass:ergodic_input}--\ref{ass:cross_covariance_zero} hold, and past input is used as instrumental variables. See, e.g., \cite{soderstrom2002instrumental} for the analysis of consistency using instrumental variables.
	\end{rem}
}

\section{FAULT SYSTEM IDENTIFICATION}\label{sec:fault_id}

The next task is to identify $\Fm, \Gm,$ and reconstruct the fault signal $\vv$ and state trajectory $\xv$. The latter is not available from the identification of the previous section as it relies on \eqref{eq:system_fault_output}, thus recovering $\tilde{\xv}$. Hereafter, we assume that $\Am, \Bm, \Cm, \Dm$ are known; thus, $\Om$ and $\Tm$ are also available; and for analysis purposes we consider that noise is absent, i.e., $\Wm = \O.$ In this case, referring to \eqref{eq:data_equation}, denote $\Rm \coloneqq \Ym - \Tm\Um$ as the Hankel matrix of the fault residual, and we get
\begin{equation}\label{eq:blind_data}
    \Rm_{s,T} = \Om_{s}\Xm_{T} + \Tm^f_{s}\Vm_{s,T},
\end{equation}
where the unknowns are $\Xm_{T}$, $\Tm^f_{s}$ and $\Vm_{s,T}$, or, referring to the state-space equation  \eqref{eq:residual_system}, to obtain the matrices $\Fm, \Gm,$ the initial state $\xiv_0$ and the trajectory $\vv$.
The solution to this problem is clearly non-unique. E.g., if $(\xiv_0, \Fm, \Gm, \vv)$ is a solution to the estimation problem, so is $(\xiv_0, \Fm\Jm, \Gm\Jm, \Jm^{-1}\vv)$ for any invertible matrix $\Jm \in \R^{n_f \times n_f}$.  As we shall see, the presence of transmission zeros from fault to output further increases the set of possible solutions.

\extra{
Let us now investigate possible approaches to identifying the fault matrices and signals.

\subsection{Rank minimization approach}\label{ssec:rankmin}
Consider the data equation with $s=T$
\[\Rm_{T,T} = \Om_T\xiv_0 + \Tm^f_T\Vm_{T,T}.\]
Let $\Tm_x$ be the Toeplitz matrix associated with the system $(\Am,\m{\I & \O},\Cm,\m{\O & \I}),$ for which a consistent estimator is available. We have that
\[\Tm^f\Vm_{T,T} = \Tm_x \m{\m{\Fm \\ \Gm}\vv(1) \\ \vdots \\ \m{\Fm \\ \Gm}\vv(T)} \eqqcolon \Tm_x\Zm_{T,T}\]
and, as such
\[\Zm_{1,T} = \m{\Fm \\ \Gm}\Vm_{1,T}.\]
By assumption of observability, $\Tm_x$ is always full row rank. Hence, if no constraint is made on the dimension of the unknown input space $\nv$, a trivial solution such as $\bsm{\Fm \\ \Gm} = \I$ is automatic. Moreover, if there exists a solution to \eqref{eq:blind_data} with a given $\nv$, then there is a solution with unknown input dimension $\nv+1$. Therefore, a natural requirement of minimality is given on the dimension of the input space; in other words, we want to find $\Fm,\Gm$ with the least number of columns. Because $\rank(\Wm_1) \leq \nv$, we can achieve that by the following rank minimization problem:
\begin{equation}\label{eq:rankmin}
\begin{aligned}
    \min_{\zv, \xiv_0}\quad & \rank(\Zm_1) \\
    \text{subject to}\quad & \Rm_{T,T} = \Om_T\xiv_0 + \Tm_x\Zm_{T,T}.
\end{aligned}
\end{equation}
The solution to \eqref{eq:rankmin} is the minimal input-space dimension that explains $\rv$. Next, take $\xiv_0, \zv$ as the minimizers of  \eqref{eq:rankmin}, and let $\Um_Z\Sigma_Z\Vm_Z\tran = \Zm_{1,T} $ be the singular value decomposition of $\Zm_{1,T}$. Set $\bsm{\Fm \\ \Gm} = \Um_Z$ and $\Vm_{1,T} = \Sigma_Z\Vm_Z\tran$, and build the matrices $\Tm^f_T$ and $\Vm_{T,T}$.  Then, $\Rm_{T,T} = \Om\xiv_0 + \Tm^f_T\Vm_{T,T}$.

Since rank minimization problems are generally NP-hard  \cite{recht2010guaranteed}, we convexify \eqref{eq:rankmin} by replacing rank with the nuclear norm, which in some conditions may yield a rank-minimal solution and has been applied for different problems in blind identification \cite{noom2023data, noom2024proximal}. 
The relaxed problem is
\begin{equation}\label{eq:nucnorm}
\begin{aligned}
    \min_{\wv, \xiv_0}\quad & \|\Wm_1\|_* \\
    \text{subject to}\quad & \Rm_{T} = \Om_T\xiv_0 + \Tm_x\wv.
\end{aligned}
\end{equation}
The relaxation above is not guaranteed to provide a minimal rank solution, thus we investigate whether it can be solved exactly by other means.
}

\subsection{A notion of output behavioral equivalence}\label{ssec:beh}

System \eqref{eq:residual_system} can be described by
\[ \Hm(\q)\m{\xiv \\ \vv} = \Lm\rv, \quad \Hm(\q) \coloneqq \m{\Am - \q\I & \Fm \\ \Cm & \Gm}, \quad \Lm = \m{\O \\ \I},\]
where $\q$ is the left-shift operator. %
Using a behavioral set description \cite{willems1991paradigms} helps to reveal what can be distinguished from the unknown input set given the residual and the available matrices $\Am$ and $\Cm.$  We call the \emph{output behavior set} $\Bs$ of a system with unknown initial state and input the set of all trajectories $\rv$ for which there exist initial state $\xiv_0$ and input trajectory $\vv$ that can generate $\rv$. The follow descriptions are equivalent:
\begin{subequations}\label{eq:behavior}
    \begin{align}
    \Bs = \left\{\rv: \N \to \R^m\ \mid \exists \xiv, \vv: \Hm(\q)\m{\xiv \\ \vv} = \Lm\rv \right\} \label{eq:behavior_state}\\
    = \{\rv: \N \to \R^m\ \mid \Nm(\q)\rv = \O \}, \label{eq:behavior_null}
\end{align}
\end{subequations}
where the latter is known as the kernel representation, and $\Nm(\q) = \Hm(\q)\kernel\Lm,$ where $\Hm(\q)\kernel$ is a minimal polynomial basis for the left null space of $\Hm(\q)$. Then kernel representation reveals the following fact:

\begin{prop}\label{prop:trivial_behavior} If Rosenbrock matrix $\Hm(\q)$ is full row rank, the behavior set $\Bs$ spans the whole space $\N \to \R^{n_y}.$    
\end{prop}

Now, suppose that $\ny=\nv$ so that $\Hm(\q)$ is square. If $\Hm(\q)$ has full normal rank, Prop.~\ref{prop:trivial_behavior} tells that the system's output behavior is full and there is not much to retrieve from it. In fact, the matrices $\Fm = \O$ and $\Gm = \I$ are a valid solution to the identification problem, with corresponding fault signal satisfying $\vv = \rv - \yv.$ Thus, we proceed this note with the tacit assumption that $\ny > \nv$, for which nontrivial results can be retrieved. In general, we are concerned with find matrices $\Fm', \Gm'$ for which the same output behavior as $\Bs$ can be generated. Following similar definitions from computer science and hybrid systems theory, see, e.g., \cite{baier2008principles, tabuada2009verification}:
\begin{defn}[Output behavior equivalence]
	Two systems with corresponding output behaviors $\Bs$ and $\Bs'$ are said to be \emph{output behaviorally equivalent} if $\Bs = \Bs'$.
\end{defn}

The behavior restricted from time $0$ to time $s$ is denoted by $\Bs|_0^s$.  With an abuse of notation, we take an element of $\Bs|_0^s$ as the column vector obtained by stacking $\rv(0), \rv(1), ...$ vertically. Clearly, $\Bs|_0^s = \Rs\left(\m{\Om_s & \Tm_s}\right).$ What is less clear but important is the following fact:
\begin{lem}\label{lem:behavioral_extension}
	Consider a system $(\Am,\Fm,\Cm,\Gm)$ of order $\nx$. Assume $s\geq \nx$ and let $\rv_{0:s} \coloneqq \m{\rv(0)\tran & \rv(1)\tran & \cdots & \rv(s)\tran}{}\tran \in \Bs|_0^s$. Then, $\rv_{0:s+1} \in \Bs|_0^{s+1}$ if and only if $\rv_{1:s+1} \in \Bs|_0^{s}$.
\end{lem}

\begin{proof}
	{\setlength\arraycolsep{3pt}
	To prove the forward implication, take $\xiv_0, \vv_{0:s+1}$ such that $\rv_{0:s+1} = \m{\Om_{s+1} & \Tm_{s+1}}\bsm{\xiv_0 \\ \vv_{0:s+1}}.$ Let $\xiv_1 \coloneqq \Am\xiv_0 + \Bm\vv_0.$ Then, $\rv_{1:s+1} = \m{\Om_{s} & \Tm_{s}}\bsm{\xiv_1 \\ \vv_{1:s+1}}$.
	To prove the backward implication, assume that $\rv_{0:s+1} \notin \Rs\left(\m{\Om_{s+1} & \Tm_{s+1}}\right)$. Because
	\[ \m{\Om_{s+1} & \Tm_{s+1}} = \mm{6}{\Om_s & \vline & & \Tm_s & & \O \\ \hline \Cm\Am^{s+1} & \vline & \Cm\Am^s\Fm & \cdots & \Cm\Fm & \Gm},\]
	then $\rv_{0:s} \in \Rs\left(\m{\Om_s & \Tm_s}\right)$ implies that $\rv_{s+1} \notin \Rs\left(\m{\Cm\Am^{s+1} & \vline & \Cm\Am^s\Fm & \cdots & \Cm\Fm & \Gm}\right).$ However, we have by Lemma \ref{lem:stability_ranges} (see Appendix) that $\Rs\left(\Cm\Am^{s+1}\right) = \Rs\left(\Cm\Am^{s}\right)$. In addition, by Cayley--Hamilton, $s\geq \nx$ implies that $\Am^s$ is a linear combination of lower powers of $\Am$. As such, the column $\Cm\Am^s\Fm$ does not add to the column space of 
		\(\Rs\left(\m{\Cm\Am^s\Fm & \cdots & \Cm\Fm & \Gm}\right) = \Rs\left(\m{\Cm\Am^{s-1}\Fm & \cdots & \Cm\Fm & \Gm}\right).\)
	Therefore, $\rv_{s+1} \notin \Rs\left(\m{\Cm\Am^{s} & \vline & \Cm\Am^{s-1}\Fm & \cdots & \Cm\Fm & \Gm}\right),$ and $\rv_{1:s+1} \notin \Bs|_0^{s}$.
	}
\end{proof}

By induction on Lemma \ref{lem:behavioral_extension}, we have that $\Bs|_0^{\nx}$ fully characterizes the output behavior of a system:
\begin{thm}\label{thm:behavioral_equivalence}
	Consider two systems $(\Am,\Fm,\Cm,\Gm)$ and  $(\Am',\Fm',\Cm',\Gm')$ of order $\nx$, with behaviors $\Bs$ and $\Bs'$, respectively. If $\Bs|_0^\nx = \Bs'|_0^\nx,$ then they are output behaviorally equivalent.
\end{thm}

\subsection{Exact approach for fault system identification}\label{ssec:fault_id}

We start by invoking the definition of left invertibility:
\begin{defn}[Left-invertibility~\cite{kirtikar2011delay}]
	The system $(\Am,\Bm,$ $\Cm,\Dm)$ with transfer function $\mathbf{G}(\z):=\Cm(\z\I-\Am)^{-1}\Bm+\Dm$ is said to be \textit{$l$-delay left invertible} if there exists an $\nup \times \ny$ proper transfer function $\mathbf{G}^{\dagger}(\z)$ such that $\mathbf{G}^{\dagger}(\z)\mathbf{G}(\z)=\z^{-l}\I$ for almost all $\z\in\mathbb{C}$. If one such $l\geq0$ exists, the system is simply called \emph{left invertible}.
\end{defn}

Let us rewrite \eqref{eq:blind_data} as
\begin{equation}\label{eq:qr_decomp}
	 \Rm_s = \m{\Om_s & \Tm_f}\m{\Xm_s \\ \Vm_s} = \Qm\Zm,
\end{equation}
where $\Qm$ is an orthonormal basis for the column space of $\Rm_s$ (obtained through, e.g., QR factorization.) If $s$ is sufficiently large but still $\Rm_s$ has more columns than rows, then under the assumption that $\Vm_s$ is p.e.~of sufficiently high order, the rank of $\Rm_s$ is equal to the rank of $\m{\Om_s & \Tm_f}$. As recently proved in \cite{sanjeevini2020counting}, the rank of the latter matrix is reduced by the presence of zeros in the system. More precisely, for left invertible systems, the rank deficiency of $\m{\Om_s & \Tm_f}$ is precisely the number of transmission zeros: %
\begin{thm}[Adapted from \cite{sanjeevini2020counting}]\label{thm:rankzeros} Let $l$ be the smallest integer such that the minimal system $(\Am,\Bm,\Cm,\Dm)$ is $l$-delay left invertible, and $\zeta$ be the number of finite and infinite transmission zeros counting multiplicity. Then, if $s \geq \max(l,\nx),$ it holds that $\rank{\m{\Om_s & \Tm_s}}= \nx+s\nup - \zeta.$
\end{thm}
\begin{proof}
    Let $\zeta_f$ be the number of finite zeros counting multiplicity and $\zeta_i$ be the number of infinite zeros. Theorem III.4 in \cite{sanjeevini2020counting} states that that, if $l \geq n_x$, then 
    \( \nx - \rank\left(\m{\Om_s & \Tm_s}\right) + \rank(\Tm_s) = \zeta_f. \)
    Theorem IV.8 in \cite{sanjeevini2020counting} gives that, if $s \geq l$, then $\rank(\Tm_s) = s\nup- \zeta.$ The desired result is obtained by combining these two expressions.
\end{proof}

In light of Theorem \ref{thm:rankzeros}, we shall assume the following:
\begin{assum}\label{assum:main}  
	The residual system \eqref{eq:residual_system} is $l$-delay left-invertible, $(\Am,\Cm)$ is observable, the matrix $\bsm{\Xm \\ \Vm_s}$ is full row rank, and  $s \geq \max(l,n_x-1).$
\end{assum}

The role of these assumptions is made clearer after the following result.

\begin{prop}\label{prop:rank_u}
    If Assumption \ref{assum:main} holds, then all of the following is true:
    \begin{enumerate}[(i)]
    	\item the column spaces of $\Rm_s$ and $\m{\Om_s & \Tm_s}$ are equal;
    	\item $\rank{\Rm_s} = n_x + s\nv - \zeta;$
    	\item $\nv = \rank{\Rm_{s+1}} - \rank{\Rm_s}.$
    \end{enumerate}
\end{prop}
\begin{proof}
    Result (i) is trivial from \eqref{eq:qr_decomp} since  $\bsm{\Xm \\ \Vm_s}$ is full row rank by assumption. Hence, $\rank{\Rm_s} =  \rank\m{\Om_s & \Tm_s}$. For (ii), Theorem \ref{thm:rankzeros} in (i) gives $\rank{\Rm_s} =  \rank\m{\Om_s & \Tm_s} = n_x + s\nv - \zeta.$ Finally, we use (ii) to get that  $\rank{\Rm_{s+1}} - \rank{\Rm_s} = n_x + (s+1)\nv - \zeta - n_x - s\nv + \zeta = \nv$, obtaining (iii).
\end{proof}

\begin{rem}[Left invertibility is general]
	The assumption that system \eqref{eq:residual_system} is left-invertible is compatible with the problem of recovering the minimal input dimension that explains data. Clearly, if the system is not left invertible, then there exists combinations of fewer inputs that can produce the same output. Hence, $\nv$ in Proposition \ref{prop:rank_u} is the minimal input dimension that explains the data.
\end{rem}

\begin{rem}[Choosing $s$]
	While $n_x$ is available from the identification of $\Am$, $l$ is a difficult number to know beforehand. A safe choice is to pick $s \geq n_x$, since, from \cite[Prop.~1]{kirtikar2011delay}, any left-invertible system is $n_x$-left invertible.
\end{rem}

\begin{rem}[Persistency of excitation by $\vv$]
	Assuming that $\bsm{\Xm \\ \Vm_s}$ is full rank is a combination of open-loop condition, which was already placed in Section \ref{sec:prob}, and a persistency of excitation assumption, which is relatively standard in subspace identification. %
	In open loop, $\Xm$ and $\Vm$ have generically disjoint row spaces. Since this is the residual subsystem, we can freely initialize $\xiv(0)$\altextra{, and a generic non-zero initial state renders $\Xm$ full row rank provided $\Am$ does not have eigenvalues with geometric multiplicity bigger than 1.}{.} Finally, because we have no control on the signal $\vv$, $\Vm$ being full rank may be regarded as an excessively stringent condition; however, if one is monitoring the system before the onset of the fault, the time window may be picked in such a way that it holds. \extra{E.g., consider a uni-dimensional fault $v$ starting at time $s-1$: the matrix \[\m{0 & \cdots & v(s-1) & \cdots \\ \vdots & \iddots  \\ v(s-1) & \cdots}\] is always full row rank.}
\end{rem}

Hereafter, we explain the procedure to recover the pair $(\Fm, \Gm)$  under Assumption \ref{assum:main}. From \eqref{eq:qr_decomp}, we get
\[  \Rm_s{\m{\Xm_s \\ \Vm_s}}^{-R} = \m{\Om_s & \Tm_f} = \Qm\Zm{\m{\Xm_s \\ \Vm_s}}^{-R} \coloneqq \Qm\tilde{\Zm}. \]
Thus, partitioning $\tilde{\Zm} = \m{\Zm_x & \Zm_u}$ where the first block has $\nx$ columns, we have that
\(\Tm_f = \Qm\Zm_u.\)
Our problem amounts to finding a full-column rank matrix $\Zm_u$ such that
\begin{equation}
   \Qm\Zm_u = \m{\Gm & \O & \O & \O& \cdots & \O \\
   	                                      \Cm\Fm & \Gm & \O  & \O & \cdots & \O \\
   	                                      \Cm\Am\Fm & \Cm\Fm & \Gm & \O & \cdots & \O \\
   	                                      \vdots & \vdots & & \ddots \\
   	                                      \Cm\Am^{s-2}\Fm & \Cm\Am^{s-3}\Fm & & \cdots & & \Gm}.
\end{equation}
Partitioning conformally $\Qm$ by $s$ block rows $\Qm_i$ and $\Zm_u$ by $s$ block columns $\Zm_j$, this equality can be rearranged as
\begin{equation}\label{eq:rearranged}
	\begin{aligned}
		\Qm_i\Zm_j &= \O,              &&  \text{for all } i,j \in \{1,2,...,s\}: j > i, \\
		\Qm_i\Zm_j &=\Qm_{i+1}\Zm_{j+1}, &&    \text{for all } i,j \in \{1,2,...,s\}: j \leq i ,\\
		\Om_{s-1}\hat{\Fm} &= \Qm_{2:s}\Zm_1.
	\end{aligned}
\end{equation}
Since all matrices multiplying on the left are given, \eqref{eq:rearranged} can be rewritten as the nullspace problem \(	\Mm\m{\Zm_1\tran & \cdots & \Zm_s\tran & \hat{\Fm}{}\tran}{}\tran = \O, \)
where $\Mm$ is built following the equations in \eqref{eq:rearranged}. By construction, $\hat{\Fm}$ contains the column space of $\Fm$, and $\hat{\Gm} = \Qm_1\Zm_1$ contains that of $\Gm$.  With standard linear algebraic arguments, we arrive at the main result of this note.
\begin{thm}\label{thm:FG}
	Let Assumption \ref{assum:main} hold, and further take $s \geq \nx$. Let $\Zm_j \in \R^{(\nx + \ny)\times \nz}$, $j \in \{1,2,...,s\}$ and $\hat{\Fm}  \in \R^{\nx\times \nz}$ be maximal in the sense that their columns form a basis for the solution of \eqref{eq:rearranged}, and let $\hat{\Gm} \coloneqq \Qm_1\Zm_1 \in \R^{\ny\times \nz}$. Then,
	\begin{enumerate}[(i)]
	\item There exists a matrix $\Pm \in \R^{\nz \times \nv}$ such that $\Fm = \hat{\Fm}\Pm$ and $\Gm = \hat{\Gm}\Pm$;
	\item For almost any full column-rank $\Pm \in \R^{\nz \times \nv}$ , the system $(\Am,\hat{\Fm}\Pm,\Cm,\hat{\Gm}\Pm)$ is output behaviorally equivalent to $(\Am,\Fm,\Cm,\Gm)$;
	\item If $(\Am,\Fm',\Cm,\Gm')$ is output behaviorally equivalent to $(\Am,\Fm,\Cm,\Gm)$, then there exists a matrix $\Pm'$ such that $\Fm' = \hat{\Fm}\Pm$ and $\Gm' = \hat{\Gm}\Pm$.
	\end{enumerate}
\end{thm}
\begin{proof}
	Result (i) is obtained by construction of $\hat{\Fm}.$ For (ii), note that the matrices $\Om_s'$ and $\Tm_s'$ for system $(\Am,\hat{\Fm}\Pm,\Cm,\hat{\Gm}\Pm)$ are obtained via
	{\setlength\arraycolsep{3pt}
	\[\m{\Om'_s & \Tm'_s} = \m{\Cm  & \hat{\Gm} \\ \Cm\Am & \Cm\hat{\Fm} & \hat{\Gm} \\ \vdots & \vdots & \ddots & \ddots \\ \Cm\Am^{s} & \Cm\Am^{s-1}\hat{\Fm} & \cdots & \Cm\hat{\Fm} & \hat{\Gm}}
	\m{\I \\ & \Pm \\ & & \ddots \\ & & & \Pm}.\]
    }%
    Since for one such $\Pm$ the resulting matrix has the same column space as $\m{\Om'_s & \Tm'_s},$ for almost all $\Pm$ this holds. Therefore, the corresponding restricted behaviors satisfy $\Bs|_0^s = \Bs'|_0^s$ and, by Theorem \ref{thm:behavioral_equivalence}, $\Bs = \Bs'$. Finally, (iii) is trivial by Theorem \ref{thm:behavioral_equivalence} and the construction of $\hat{\Fm}, \hat{\Gm}$.
\end{proof}
Theorem \ref{thm:FG} not only provides a method to recover matrices $\Fm',\Gm'$ from the residual system, it is also a realization algorithm. Given $(\Am,\Fm,\Cm,\Gm)$ we can use Theorem \ref{thm:FG} to generate any output behaviorally equivalent system of the form $(\Am,\Fm',\Cm,\Gm')$ by simulating $(\Am,\Fm,\Cm,\Gm)$ such that Assumption \ref{assum:main} holds, and performing the steps outlined above.

{\bf Recovering $\xv_0$ and $\vv$.} Once $\Fm',\Gm'$ are determined via Theorem \ref{thm:FG}, $\xiv_0$ and $\vv$ from \eqref{eq:residual_system} can be determined as any solution that satisfies $\rv = \Om_T\xiv_0 + \Tm'_T\vv$, which by construction is sure to exist, but is not unique if $(\Am,\Fm,\Cm,\Gm)$ has invariant zeros \cite{kirtikar2011delay}. Any existing method in the literature may be used, with different performances depending on the location of the transmission or invariant zeros. With $\tilde{\xv}_0$ recovered from the identification step in Section \ref{sec:id}, the full state trajectory $\xv = \xiv + \tilde{\xv}$ can be constructed.

%% file: numerical.tex
\section{NUMERICAL EXAMPLE{\footnotemark}}
\footnotetext{{To reproduce these results, see \url{https://github.com/ggleizer/fault_id}.}}
We consider the system of the form \eqref{eq:system} with
\begin{align*}
	\Am &= \m{0 & 1 & 0 \\ 0 & 0 & 1 \\ -0.25 & 0.75 & 0.25}\!, &\Bm &= \m{0\\0\\1}\!, &\Fm&=\m{0.938\\	0.328\\	0.115}\!,\\
	\Cm &= \m{1 & 0 & 0 \\ 0 & 1 & 0}\!, &\Dm &=\m{0\\0}\!, &\Gm&=\m{0 \\ 0}\!.
\end{align*}
The pair $(\Am,\Fm)$ is not controllable, and the residual system $(\Am,\Fm,\Cm,\Gm)$ has one transmission zero at infinity. 
We applied a white Gaussian pseudo-random input $\uv$ for 1000 time units, with random initial condition, $\vv(k) = 0.1+\sin(0.25k^{1.3})$ and no noise. Using PI-MOESP as outlined in Section \ref{sec:id}, we retrieved the matrices (in canonical observable form)
\begin{align*}
	\hat{\Am} &= \m{0 & 1 & 0 \\ 0 & 0 & 1 \\ -0.256  & 0.749 &  0.261}\!, &\hat\Bm &= \m{-0.004\\-0.023\\1.015}\!,\\
	\hat\Cm &= \m{1 & 0 & 0 \\ -0.006 &  0.978 &  0.0126}\!, &\hat\Dm &=\m{0.03\\0.016}\!.
\end{align*}
To recover $\Fm,\Gm$, first we considered the actual matrices $\Am,\Bm,\Cm,\Dm$ to generated the residual signal. \altextra{We applied the rank minimization approach o Section \ref{ssec:rankmin} to recover $\Fm,\Gm$. The solver took over 4 minutes, but returned a rank-2 solution, with recovered matrix $\bsm{\hat\Fm \\ \hat\Gm}$ not spanning $\bsm{\Fm \\ \Gm}$. In contrast, the approach}{The approach} outlined in Section \ref{ssec:fault_id} identified that $\nv = 1$ via Prop.~\ref{prop:rank_u} with $s=5$, and using Theorem \ref{thm:FG} {we obtained $\hat{\Fm}_1, \hat{\Gm}_1$ in \eqref{eq:recovered_fg}:
\begin{equation}\label{eq:recovered_fg}
	\m{\hat{\Fm}_1 \\ \hat{\Gm}_1} = \mm{2}{0 & 0.938 \\ 0 & 0.328 \\ 0	& 0.115 \\ \hline -0.944	& 0 \\ -0.33 & 0}, ~~~\m{\hat{\Fm}_2 \\ \hat{\Gm}_2} = \mm{2}{0 & 0.933 \\ 0 & 0.338 \\ 0	& 0.122 \\ \hline   -0.944	& 0 \\ -0.33 & 0},
\end{equation}%
}
that is, the second column is precisely $\bsm{\Fm \\ \Gm}$, while the first column gives an output behaviorally equivalent system. Notice that, if we had prior knowledge that the fault was not on the measurements, we could recover $\Fm,\Gm$ precisely.

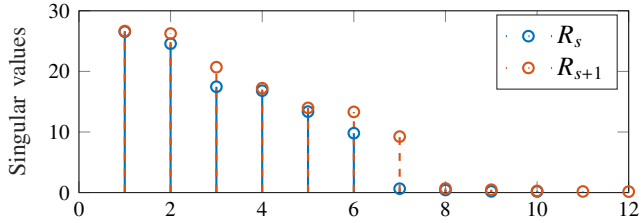
\begin{figure}
	\centering
	\input{svs.tex}
	\caption{Singular values of Hankel residual matrices $\Rm_5$ and $\Rm_{6}$ when using identified system to generate the residual. Clearly, only the seventh singular value of $\Rm_5$ is negligible compared to that of $\Rm_{6}$, indicating that $\nv=1$.}
	\label{fig:svs}
\end{figure}

Next, we consider the identified matrices $\hat\Am,\hat\Bm,\hat\Cm,\hat\Dm$ to generate the residual. Now, because the residual is no longer purely affected by the fault, but also by modeling mismatch, the rank condition of Prop.~\ref{prop:rank_u} is too frail to retrieve $\nv$. Instead, we inspect the singular values of $\Rm_s$ and $\Rm_{s+1}$ to see the difference in significant singular values, see Fig.~\ref{fig:svs}. Selecting a threshold for the singular values, we used the methods of Section \ref{ssec:fault_id}, retrieving {$\hat{\Fm}_2, \hat{\Gm}_2$ in \eqref{eq:recovered_fg},}
whose second column is within 1.4\% of the original $\hat\Fm$.

{\bf Monte-Carlo analysis.} To assess the performance of reconstructing $\Fm, \Gm$ from noisy data, we conduct Monte-Carlo simulation where 100 stable systems of the form \eqref{eq:system} were generated with $\nx=5, \nup=1, \ny=3, \nv=2$. Out of them, 25 have no transmission zeros, 25 have 1, 25 have 2, and 25 have 3. All transmission zeros were real and placed randomly via a normal distribution, leaving several systems non-minimum phase from fault to output. The input, initial state, and fault signals were selected as in the previous example, with $v_1(k) = 0.1+\sin(0.25k^{1.3})$ and $v_2(k) = 1 - 0.99^k + z(k)$, where $z(k)$ is a zero-mean i.i.d. Gaussian noise with unit variance. On top of that, colored measurement noise was introduced to the outputs with variance tuned to achieve an SNR of 40 dB. 

Given the subspace nature of fault matrices, we used for error metric the Grassmannian distance \cite{ye2016schubert} between $\m{\Fm\tran & \Gm\tran}\tran$ and $\m{\hat{\Fm}\tran & \hat{\Gm}\tran}\tran,$ normalized by the maximal Grassmannian distance possible in the ambient space, so that 0\% corresponds to equality and 100\% to orthogonal subspaces. The median error across all examples was 0.82\%. Box plots of the obtained errors can be seen in Fig.~\ref{fig:box_nz} as a function of the number of transmission zeros. When there were no transmission zeros, all recovered matrices were very close to perfect; once transmission zeros were present, while the majority of errors were still very small as seen by the median lines, a few cases observed large errors.

\begin{figure}
	\centering
	\input{figures/box_nz.tex}
	\caption{Grassmanian error between real and recovered fault matrices in the Monte-Carlo simulation.}
	\label{fig:box_nz}
\end{figure}
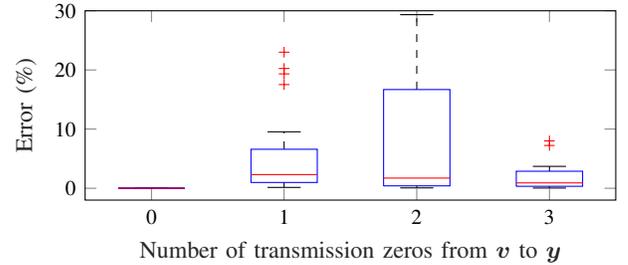

%% file: svs.tex
%
%
\definecolor{mycolor1}{rgb}{0.00000,0.44700,0.74100}%
\definecolor{mycolor2}{rgb}{0.85000,0.32500,0.09800}%
\begin{tikzpicture}
	
\pgfplotsset{every tick label/.append style={font=\footnotesize}}

\begin{axis}[%
width=0.5\textwidth,
height=4cm,
xmin=0,
xmax=12,
ymin=0,
ymax=30,
ylabel style={font=\color{white!15!black} \small},
ylabel={Singular values},
axis background/.style={fill=white},
legend style={legend cell align=left, align=left, draw=white!15!black}
]
\addplot[thick, ycomb, color=mycolor1, mark=o, mark options={solid, mycolor1}] table[row sep=crcr] {%
1	26.5500789412737\\
2	24.5755309469592\\
3	17.464565160419\\
4	16.8178553234019\\
5	13.3593909879432\\
6	9.80428541092819\\
7	0.642137996198499\\
8	0.424179989481712\\
9	0.196427980417734\\
10	0.1811099220302\\
};
\addplot[forget plot, color=white!15!black] table[row sep=crcr] {%
0	0\\
12	0\\
};
\addlegendentry{$R_s$}

\addplot[ycomb, thick, dashed, color=mycolor2, mark=o, mark options={solid, mycolor2}] table[row sep=crcr] {%
1	26.6613566971983\\
2	26.2480832086168\\
3	20.6959018289427\\
4	17.2298675859753\\
5	13.9731111147254\\
6	13.3147972062507\\
7	9.23786667136778\\
8	0.674689402948074\\
9	0.473886886786674\\
10	0.271201768646024\\
11	0.184044483183736\\
12	0.157473233749879\\
};
\addplot[forget plot, color=white!15!black] table[row sep=crcr] {%
0	0\\
12	0\\
};
\addlegendentry{$R_{s+1}$}

\end{axis}
\end{tikzpicture}%

%% file: figures/box_nz.tex
%
%
\begin{tikzpicture}
	
\pgfplotsset{every tick label/.append style={font=\footnotesize}}

\begin{axis}[%
width=\linewidth,
height=4.1cm,
unbounded coords=jump,
xmin=0.5,
xmax=4.5,
xtick={1,2,3,4},
xticklabels={{0},{1},{2},{3}},
xlabel={Number of transmission zeros from $\vv$ to $\yv$},
ymin=-2,
ymax=30,
xlabel style={font=\color{white!15!black} \small},
ylabel style={font=\color{white!15!black} \small},
ylabel={Error (\%)},
axis background/.style={fill=white}
]
\addplot [color=black, dashed, forget plot]
  table[row sep=crcr]{%
1	8.2154407415061e-07\\
1	1.25492930217919e-06\\
};
\addplot [color=black, dashed, forget plot]
  table[row sep=crcr]{%
2	6.59425477262639\\
2	9.54464722446597\\
};
\addplot [color=black, dashed, forget plot]
  table[row sep=crcr]{%
3	16.6856508996249\\
3	29.3434130952109\\
};
\addplot [color=black, dashed, forget plot]
  table[row sep=crcr]{%
4	2.87754107743578\\
4	3.70170134874619\\
};
\addplot [color=black, dashed, forget plot]
  table[row sep=crcr]{%
1	0\\
1	0\\
};
\addplot [color=black, dashed, forget plot]
  table[row sep=crcr]{%
2	0.143002117051297\\
2	0.978712257046016\\
};
\addplot [color=black, dashed, forget plot]
  table[row sep=crcr]{%
3	0.0606432183868971\\
3	0.404879355000085\\
};
\addplot [color=black, dashed, forget plot]
  table[row sep=crcr]{%
4	0.0372813441184505\\
4	0.333512498194564\\
};
\addplot [color=black, forget plot]
  table[row sep=crcr]{%
0.875	1.25492930217919e-06\\
1.125	1.25492930217919e-06\\
};
\addplot [color=black, forget plot]
  table[row sep=crcr]{%
1.875	9.54464722446597\\
2.125	9.54464722446597\\
};
\addplot [color=black, forget plot]
  table[row sep=crcr]{%
2.875	29.3434130952109\\
3.125	29.3434130952109\\
};
\addplot [color=black, forget plot]
  table[row sep=crcr]{%
3.875	3.70170134874619\\
4.125	3.70170134874619\\
};
\addplot [color=black, forget plot]
  table[row sep=crcr]{%
0.875	0\\
1.125	0\\
};
\addplot [color=black, forget plot]
  table[row sep=crcr]{%
1.875	0.143002117051297\\
2.125	0.143002117051297\\
};
\addplot [color=black, forget plot]
  table[row sep=crcr]{%
2.875	0.0606432183868971\\
3.125	0.0606432183868971\\
};
\addplot [color=black, forget plot]
  table[row sep=crcr]{%
3.875	0.0372813441184505\\
4.125	0.0372813441184505\\
};
\addplot [color=blue, forget plot]
  table[row sep=crcr]{%
0.75	0\\
0.75	8.2154407415061e-07\\
1.25	8.2154407415061e-07\\
1.25	0\\
0.75	0\\
};
\addplot [color=blue, forget plot]
  table[row sep=crcr]{%
1.75	0.978712257046016\\
1.75	6.59425477262639\\
2.25	6.59425477262639\\
2.25	0.978712257046016\\
1.75	0.978712257046016\\
};
\addplot [color=blue, forget plot]
  table[row sep=crcr]{%
2.75	0.404879355000085\\
2.75	16.6856508996249\\
3.25	16.6856508996249\\
3.25	0.404879355000085\\
2.75	0.404879355000085\\
};
\addplot [color=blue, forget plot]
  table[row sep=crcr]{%
3.75	0.333512498194564\\
3.75	2.87754107743578\\
4.25	2.87754107743578\\
4.25	0.333512498194564\\
3.75	0.333512498194564\\
};
\addplot [color=red, forget plot]
  table[row sep=crcr]{%
0.75	4.74318692361997e-07\\
1.25	4.74318692361997e-07\\
};
\addplot [color=red, forget plot]
  table[row sep=crcr]{%
1.75	2.30829962455157\\
2.25	2.30829962455157\\
};
\addplot [color=red, forget plot]
  table[row sep=crcr]{%
2.75	1.72115621853197\\
3.25	1.72115621853197\\
};
\addplot [color=red, forget plot]
  table[row sep=crcr]{%
3.75	0.907293505423315\\
4.25	0.907293505423315\\
};
\addplot [color=black, only marks, mark=+, mark options={solid, draw=red}, forget plot]
  table[row sep=crcr]{%
2	17.5306445311511\\
2	19.3339286488376\\
2	20.2498452928214\\
2	23.0024849218859\\
};
\addplot [color=black, only marks, mark=+, mark options={solid, draw=red}, forget plot]
  table[row sep=crcr]{%
4	7.23259932015859\\
4	8.01194368576956\\
};
\end{axis}
\end{tikzpicture}%

%% file: appendix.tex
\section*{APPENDIX}

\begin{lem}\label{lem:stability_ranges}
	For any $\Am \in \R^{n \times n}$, $\Rs(\Am^{n+1}) = \Rs(\Am^n)$.
\end{lem}
\begin{proof}
	First, note that for any $k\geq 1$, $\Rs(\Am^{k}) = \Rs(\Am^{k-1}) \implies \Rs(\Am^{k+1}) = \Rs(\Am^k)$, as there exists $\Qm$ such that $\Am^{k}\Qm = \Am^{k-1}$, and, pre-multiplying by $\Am$ gives that $\Am^{k+1}\Qm = \Am^{k}$. Now, since $\Rs(\Am^{k+1}) \subseteq \Rs(\Am^k)$ for any $k$, we establish that, conversely, $\Rs(\Am^{k+1}) \subsetneq \Rs(\Am^k) \implies \Rs(\Am^{k}) \subsetneq \Rs(\Am^{k-1})$ and, consequently,  $\rank(\Am^{k+1}) < \rank(\Am^k) \implies \rank(\Am^{k}) < \rank(\Am^{k-1})$ . We now show that it is impossible that $\Rs(\Am^{n+1}) \subsetneq \Rs(\Am^n)$. If this were the case, then 
	\[ \rank(\Am^{n+1}) \leq \rank(\Am^n) - 1 \leq \cdots \leq \rank(\I) - n - 1 = -1,\]
	but ranks are never negative.
\end{proof}